\documentclass[a4paper,12pt,reqno]{amsart}
\usepackage{amsfonts}
\usepackage{amsmath}
\usepackage{amssymb}
\usepackage[a4paper]{geometry}
\usepackage{mathrsfs}
\usepackage{hyperref}
\renewcommand\eqref[1]{(\ref{#1})} 
\numberwithin{equation}{section}
\theoremstyle{plain}
\newtheorem{thm}{Theorem}[section]
\newtheorem{prop}[thm]{Proposition}
\newtheorem{cor}[thm]{Corollary}
\newtheorem{lem}[thm]{Lemma}

\theoremstyle{definition}

\newtheorem{rem}[thm]{Remark}
\newtheorem{ex}[thm]{Example}

\renewcommand{\wp}{\mathfrak S}
\newcommand{\Rn}{\mathbb R^{n}}

\begin{document}

   \title[Uncertainty relations on nilpotent Lie groups]
{Uncertainty relations on nilpotent Lie groups}

\author[Michael Ruzhansky]{Michael Ruzhansky}
\address{
  Michael Ruzhansky:
  \endgraf
  Department of Mathematics
  \endgraf
  Imperial College London
  \endgraf
  180 Queen's Gate, London SW7 2AZ
  \endgraf
  United Kingdom
  \endgraf
  {\it E-mail address} {\rm m.ruzhansky@imperial.ac.uk}
  }
\author[Durvudkhan Suragan]{Durvudkhan Suragan}
\address{
  Durvudkhan Suragan:
  \endgraf
  Institute of Mathematics and Mathematical Modelling
  \endgraf
  125 Pushkin str.
  \endgraf
  050010 Almaty
  \endgraf
  Kazakhstan
  \endgraf
  and
  \endgraf
  Department of Mathematics
  \endgraf
  Imperial College London
  \endgraf
  180 Queen's Gate, London SW7 2AZ
  \endgraf
  United Kingdom
  \endgraf
  {\it E-mail address} {\rm d.suragan@imperial.ac.uk}
  }

\thanks{The authors were supported in parts by the EPSRC
 grant EP/K039407/1 and by the Leverhulme Grant RPG-2014-02,
 as well as by the MESRK grant 5127/GF4.
 No new data was collected or generated during the course of research.}


     \begin{abstract}
     We give relations between main operators of quantum mechanics on one of most general classes of nilpotent Lie groups. Namely, we show relations between momentum and position operators as well as Euler and Coulomb potential operators  on homogeneous groups. Homogeneous group analogues of some well-known inequalities such as Hardy's inequality, Heisenberg-Kennard type and Heisen\-berg-Pauli-Weyl type uncertainty inequalities, as well as Caffarelli-Kohn-Niren\-berg inequality are derived, with best constants. The obtained relations yield new results already in the setting of both isotropic and anisotropic $\mathbb R^{n}$, and of the Heisenberg group. 
     \end{abstract}
    \dedicatory{``Mathematics, rightly viewed, possesses not only truth, but supreme beauty'' Bertrand Russell}
     \maketitle

\section{Introduction}
The uncertainty principle of Werner Heisenberg \cite{Heisenberg} forms a fundamental element of quantum mechanics. It is worth noting that his original argument, while conceptually
enlightening, was heuristic. The first rigorously proven uncertainty relation for position and momentum operators is due to Earle Kennard \cite{Kennard}. Meanwhile, more mathematical details were provided by Wolfgang Pauli and Hermann Weyl. The interpretation of uncertainty inequalities as spectral properties of differential operators is widely present in the literature starting from studies of Charles Fefferman \cite{Fef83} and \cite{FP81}.
Nowadays there is vast literature on uncertainty relations and their applications. We refer to a recent review article \cite{CBTW:review} for further discussions and references on this subject as well as to \cite{FS97} for an overview of the history and the relevance of this type of inequalities from a pure mathematical point of view.

\smallskip
The main aim of this note is to obtain uncertainty type relations on homogeneous (Lie) groups.  
The setting of homogeneous groups was developed by Folland and Stein in \cite{FS-Hardy}, in particular, to distill those results of harmonic analysis that depend only on the group and 
dilation structures. It turns out that
the class of homogeneous groups is one of most general subclasses
of nilpotent Lie groups and, in fact, it is often a working assumption when one is dealing with nilpotent Lie groups.
The Euclidean group $(\mathbb{R}^{n}; +)$, Heisenberg type groups, homogeneous Carnot groups, stratified Lie groups, graded Lie groups are all special cases of the homogeneous groups.
An example of a (nine-dimensional) nilpotent Lie group that does not allow for any compatible family of dilations was constructed by Dyer \cite{Dyer-1970}.
In particular, {\em $\mathbb{R}^{n}$ with an anisotropic dilation structure} gives an
important example of a homogeneous group, and the results of this note appear to be new already in this setting.
The language of Folland and Stein's homogeneous groups is, however, very convenient, as it allows for a perfect level of abstraction to make an exposition more clear.

\smallskip
The starting point of our analysis are the abstract position and momentum operators $\mathcal{P}$ and $\mathcal{M}$, which we assume to be linear operators, densely defined on $L^{2}$, with their domains containing $C_{0}^{\infty}$, and such that $C_{0}^{\infty}$ is an invariant subspace for them.
The main (and only) assumption in this paper is that $\mathcal{P}$ and $\mathcal{M}$ satisfy the relations
\begin{equation}\label{PMEfi}
2\,{\rm Re}\left(\mathcal{P}f\overline{i\mathcal{M}f}\right)
=(\mathcal{P}\circ(i\mathcal{M}))|f|^{2}=E|f|^{2}
\end{equation}
for all $f\in C_{0}^{\infty}$. The operator $E$ is a given operator, {\em the Euler operator} of the space (see e.g. \eqref{EQ:def-Euler}), so that the position and momentum operators give its factorisation as in the second equality in \eqref{PMEfi}.
The Euler operator $E$ characterises the homogeneity property: a differentiable function $f$ satisfies 
$$f(\lambda x)=\lambda^{\mu}f(x) \textrm{ for all } \lambda>0
\textrm{ if and only if } Ef=\mu f.$$
Interestingly, these relations are enough to derive properties of uncertainty relation type,
such as such as Heisenberg-Kennard and Heisen\-berg-Pauli-Weyl type uncertainty inequalities.
The property that $\mathcal{P}$ and $i\mathcal{M}$ factorise the Euler operator also allows one to  establish links between them and other operators such as the radial operator, the dilations generating operator, and the Coulomb potential operator, and prove some equalities and inequalities among them.

It seems to us a revealing fact that once operators $\mathcal{P}$ and $i\mathcal{M}$ factorise the Euler operator and satisfy the additional relation in the first equality in \eqref{PMEfi}, they must satisfy an uncertainty principle.

If the space is the Euclidean $\mathbb R^{n}$ with isotropic (standard) dilations, then 
the operators 
\begin{equation}\label{EQ:pmex1}
\mathcal{P}:=x  \textrm{ and } \mathcal{M}:=-i\nabla,
\end{equation}
i.e. the multiplication and the gradient (multiplied by $-i$) operators, satisfy \eqref{PMEfi}.
The same will hold on general homogeneous Lie groups, as we show in Example \ref{REM:pm}.

However, one can find other examples which satisfy \eqref{PMEfi}, for instance (see Example \ref{REM:pm1}), 
if $\{X_{j}\}$ is a basis of a Lie algebra $g$ of some homogeneous group $\mathbb{G}$ and ${\exp}_{\mathbb{G}}^{-1}(x)=\sum_{j=1}^{n}e_{j}(x)X_{j}$ (see \eqref{expo}), then the position operators can be defined as $f \mapsto e_{j}(x)f$, and the momentum operators can be defined as $f\mapsto\nu_{j}X_{j}f$, where $\nu_{j}$ is the homogeneous degree of $X_{j}$. 

\smallskip
Let us now very briefly review the main concepts of the homogeneous groups. For the general background details on homogeneous groups we refer to the book \cite{FS-Hardy} by Folland and Stein as well as to the recent monograph \cite{FR} by V. Fischer and the first named author.

If a Lie group (on $\mathbb{R}^{n}$) $\mathbb{G}$ has a property that
there exist $n$-real numbers $\nu_{1},\ldots, \nu_{n}$
such that the dilation
$$D_{\lambda}(x):=(\lambda^{\nu_{1}}x_{1},\ldots,\lambda^{\nu_{n}}x_{n}),\; D_{\lambda}:\mathbb{R}^{n}\rightarrow\mathbb{R}^{n},$$
is an automorphism of the group $\mathbb{G}$ for each $\lambda>0,$
then it is called a {\em homogeneous (Lie) group}.

A {\em homogeneous quasi-norm} on a homogeneous group $\mathbb G$ is
a continuous non-negative function
$$\mathbb{G}\ni x\mapsto |x|\in [0,\infty),$$
satisfying the properties

\begin{itemize}
\item   $|x^{-1}| = |x|$ for all $x\in \mathbb{G}$,
\item  $|D_{\lambda}(x)|=\lambda |x|$ for all
$x\in \mathbb{G}$ and $\lambda >0$,
\item  $|x|= 0$ if and only if $x=0$.
\end{itemize}

Let $dx$ denote the Haar measure on $\mathbb{G}$ and let $|S|$ denote the corresponding volume of a measurable set $S\subset \mathbb{G}$.
Then we have
\begin{equation}
|D_{\lambda}(S)|=\lambda^{Q}|S| \quad {\rm and}\quad \int_{\mathbb{G}}f(D_{\lambda}(x))
dx=\lambda^{-Q}\int_{\mathbb{G}}f(x)dx.
\end{equation}
A family of dilations of a Lie algebra $\mathfrak{g}$
is a family of linear mappings of the form
$$D_{\lambda}={\rm Exp}(A \,{\rm ln}\lambda)=\sum_{k=0}^{\infty}
\frac{1}{k!}({\rm ln}(\lambda) A)^{k},$$
where $A$ is a diagonalisable linear operator on $\mathfrak{g}$
with positive eigenvalues,
and each $D_{\lambda}$ is a morphism of the Lie algebra $\mathfrak{g}$,
that is, a linear mapping
from $\mathfrak{g}$ to itself which respects the Lie bracket:
$$\forall X,Y\in \mathfrak{g},\, \lambda>0,\;
[D_{\lambda}X, D_{\lambda}Y]=D_{\lambda}[X,Y].$$
Let us fix a basis $\{X_{1},\ldots,X_{n}\}$ of the Lie algebra $\mathfrak{g}$ of the homogeneous group $\mathbb{G}$
such that
$$AX_{k}=\nu_{k}X_{k}$$
for each $1\leq k\leq n$, so that $A$ can be taken to be
\begin{equation}\label{EQ:mA0}
A={\rm diag}(\nu_{1},\ldots,\nu_{n}).
\end{equation}
Then each $X_{k}$ is homogeneous of degree $\nu_{k}$ and also
$$
Q=\nu_{1}+\cdots+\nu_{n},
$$
which is called a homogeneous dimension of $\mathbb{G}$.  Homogeneous groups are necessarily nilpotent and hence, in particular, the exponential mapping $\exp_{\mathbb G}:\mathfrak g\to\mathbb G$ is a global diffeomorphism, where $\mathfrak{g}$ is the Lie algebra of $\mathbb{G}$.
The decomposition of ${\exp}_{\mathbb{G}}^{-1}(x)$ in the Lie algebra $\mathfrak g$ defines the vector
$$e(x)=(e_{1}(x),\ldots,e_{n}(x))$$
by the formula
\begin{equation}\label{expo}
{\exp}_{\mathbb{G}}^{-1}(x)=e(x)\cdot \nabla_{X}\equiv\sum_{j=1}^{n}e_{j}(x)X_{j},
\end{equation}
where $\nabla_{X}=(X_{1},\ldots,X_{n})$.
Alternatively, this means the equality
$$x={\exp}_{\mathbb{G}}\left(e_{1}(x)X_{1}+\ldots+e_{n}(x)X_{n}\right).$$
By homogeneity this implies
$$rx:=D_{r}(x)={\exp}_{\mathbb{G}}\left(r^{\nu_{1}}e_{1}(x)X_{1}+\ldots
+r^{\nu_{n}}e_{n}(x)X_{n}\right),$$
that is,
$$
e(rx)=(r^{\nu_{1}}e_{1}(x),\ldots,r^{\nu_{n}}e_{n}(x)).
$$
We define
\begin{equation}\label{EQ:Euler}
\mathcal{R} :=\frac{d}{dr},
\end{equation}
that is, for all $x\in \mathbb G$ 
\begin{equation}\label{dfdr}
\frac{d}{d|x|}f(x)=\mathcal{R}f(x),
\end{equation}
for each homogeneous quasi-norm $|x|$ on a homogeneous group $\mathbb G.$
Defining the Euler operator
\begin{equation}\label{EQ:def-Euler}
E=|x|\mathcal{R},
\end{equation}
it is not difficult to see that $E$ is homogeneous of order zero (see e.g. \cite{Ruzhansky-Suragan:identities}).
Since $\mathbb{G}$ is a general homogeneous group, it does not have to be stratified or
even graded. Therefore, the notion of a horizontal gradient does not make
sense, and hence it is natural to work with the operator $\mathcal{R}$.
For the horizontal versions of functional inequalities such as Hardy, Rellich, and Caffarelli-Kohn-Nirenberg inequalities we refer to \cite{Ruzhansky-Suragan:horizontal} and references therein.

We also refer to recent works of T. Ozawa et al \cite{IIO}, \cite{MOW:Hardy-Hayashi}, \cite{OY-uncertainty} for some of the abelian cases $\mathbb{G}\equiv \mathbb{R}^{n}$ of our discussions in the case of the (standard) isotropic dilations and the Euclidean distance $\|\cdot\|$. We note that {\em also in the abelian (both isotropic and anisotropic) cases of
$\mathbb{R}^{n}$ our results provide new insights in view of the arbitrariness of the homogeneous quasi-norm $|\cdot|$} which does not necessarily have to be the Euclidean norm.

In Section \ref{SEC:2} we give main results and their proofs.

\section{Main results and proofs}
\label{SEC:2}

Let us start by giving an example of position and momentum operators appearing as a special case of operators $\mathcal{P}$ and $\mathcal{M}$ used in this paper. First we give an example on general homogeneous groups, and then another example on the Heisenberg group.

\begin{ex}\label{REM:pm}
Let $\mathbb{G}$ be a homogeneous Lie group.
Let us define position and momentum operators as 
\begin{equation}\label{EQ:pmex1}
\mathcal{P}:=x,\, x\in\mathbb{G}, \textrm{ and } \mathcal{M}:=-i\nabla,
\end{equation}
where $\nabla=(\frac{\partial}{\partial x_{1}},\ldots, \frac{\partial}{\partial x_{n}})$ is an anisotropic gradient on $\mathbb{G}$ consisting of partial derivatives with respect to coordinate functions. 
Here we understand $\mathcal{P}$ as the scalar multiplication by the coordinates of the variable $x$, i.e. $\mathcal{P} v=\sum x_{j} v_{j},$ where $x_{j}$ are the coordinate functions of $x\in\mathbb{G}$, see \cite[Section 3.1.4]{FR} for the detailed discussion of these functions on homogeneous groups.

The operators $\mathcal{P}$ and $\mathcal{M}$ clearly satisfy the relation
\begin{equation}\label{EQ:pmex2}
2\,{\rm Re}\left(xf\cdot\nabla f\right)
=x\cdot\nabla|f|^{2}=E|f|^{2},
\end{equation}
where $E$ is the Euler operator from \eqref{EQ:def-Euler}, that is, 
$$E=x\cdot\nabla \textrm{ and } \mathcal{R}=\frac{x\cdot\nabla}{|x|}=\frac{d}{d|x|}.$$
Although $x_{j}$ and $\frac{\partial}{\partial x_{j}}$ may have different degrees of homogeneity depending on $j$, the Euler operator $E$ is homogeneous of order zero.
The last equality can be checked directly:
$$\frac{d}{d|x|}f(x)=\frac{d}{d|x|}f\left(\frac{x}{|x|}|x|\right)=\frac{x}{|x|}\frac{d}{dx}f(x)=\frac{x\cdot\nabla}{|x|}f(x),$$
for any $x\in\mathbb{G}$ and differentiable function $f$. Here we have used the fact that $\frac{x}{|x|}$ does not depend on $|x|$.
In the notation \eqref{EQ:pmex1} the relations \eqref{EQ:pmex2} can be expressed as
\begin{equation}\label{PMEfex3}
2\,{\rm Re}\left(\mathcal{P}f\overline{i\mathcal{M}f}\right)
=(\mathcal{P}\circ(i\mathcal{M}))|f|^{2}=E|f|^{2}.
\end{equation}
\end{ex}	

We note that the anisotropic gradient $\nabla$ can be expressed in terms of the left-invariant group gradient $\nabla_{X}=(X_{1},\ldots,X_{n})$. Such relations are well known and can we written as
$$
\frac{\partial}{\partial x_{j}}=X_{j}+\sum_{\stackrel{1\leq k\leq n}{\nu_{j}<\nu_{k}}} p_{j,k} X_{k},
$$
for some homogeneous polynomials $p_{j,k}$ on $\mathbb{G}$ of homogeneous degree $\nu_{k}-\nu_{j}>0$, see e.g. \cite[Section 3.1.5]{FR}.

\begin{ex}\label{REM:pm1}
Consider the Heisenberg group $\mathbb{H}$ on $\mathbb R^{3}$. As discussed in the introduction,
the exponential map of the group is globally invertible and its inverse map is given by the formula
\begin{equation}\label{expo}
{\exp}_{\mathbb{H}}^{-1}(x)=e(x)\cdot \nabla_{X}\equiv\sum_{j=1}^{3}e_{j}(x)X_{j},
\end{equation}
where $\nabla_{X}=(X_{1},X_{2},X_{3})$ is the full gradient of $\mathbb{H}$ with
$X_{1}=\frac{\partial}{\partial x_{1}}+2x_{2}\frac{\partial}{\partial x_{3}},$
$X_{2}=\frac{\partial}{\partial x_{2}}-2x_{1}\frac{\partial}{\partial x_{3}},$
and $X_{3}=-4\frac{\partial}{\partial x_{3}}$ as well as $e(x)=(e_{1}(x),e_{2}(x),e_{3}(x))$ with $e_{1}(x)=x_{1}$, $e_{2}(x)=x_{2}$
and $e_{3}(x)=-\frac{1}{4}x_{3}$.
In this case, the position and momentum operators can be defined as 
\begin{equation}\label{EQ:pmex1}
\mathcal{P}:=e(x),\, x\in\mathbb{G}, \textrm{ and } \mathcal{M}:=-i\nabla_{X}.
\end{equation}
It is clear that these operators satisfy the relations \eqref{PMEfex3}.
Now let us check the relation \eqref{EQ:def-Euler} between the Euler operator $E_{\mathbb{H}}:=e(x)\cdot\nabla_{X}$ and the radial operator $\mathcal{R}_{\mathbb{H}}=\frac{d}{d|x|}$:
\begin{align*}E_{\mathbb{H}}=e(x)\cdot\nabla_{X} & =x_{1}\left(\frac{\partial}{\partial x_{1}}+2x_{2}\frac{\partial}{\partial x_{3}}\right)+x_{2}\left(\frac{\partial}{\partial x_{2}}-2x_{1}\frac{\partial}{\partial x_{3}}\right)-\frac{1}{4}x_{3} \left(-4\frac{\partial}{\partial x_{3}}\right)
\\ & = x_{1}\frac{\partial}{\partial x_{1}}+x_{2}\frac{\partial}{\partial x_{2}}+
x_{3}\frac{\partial}{\partial x_{3}}
\\ & = |x|\left( \frac{x_{1}}{|x|}\frac{\partial}{\partial x_{1}}+ \frac{x_{2}}{|x|}\frac{\partial}{\partial x_{2}}+
\frac{x_{3}}{|x|}\frac{\partial}{\partial x_{3}}\right)
\\ & = |x|\frac{d}{d|x|}=|x|\mathcal{R}_{\mathbb{H}}.
\end{align*}
  
\end{ex}	

\subsection{Assumptions of this paper}

In this paper, in particular, 
we show relations between abstract position $\mathcal{P}$ and momentum $\mathcal{M}$ operators on homogeneous groups. 
These will be the operators providing a suitable factorisation for the Euler operator motivated by the relations \eqref{PMEfex3}. Although we could have worked specifically with operators
$\mathcal{P}$ and $\mathcal{M}$ from Example \ref{REM:pm}, it is good to emphasise exactly which of their properties we need to obtain the uncertainty principles and other functional relations.
However, we would like to emphasise that in the setting of homogeneous groups and already in the anisotropic $\mathbb R^{n}$ the subsequent results are new also for operators from 
Example \ref{REM:pm}, and also in the (usual) isotropic $\Rn$ in view of an arbitrary choice of a homogeneous quasi-norm $|\cdot|$.

\smallskip
Thus, from now on, let $\mathcal{P}$ and $\mathcal{M}$ be linear operators, densely defined on $L^{2}(\mathbb{G})$, with their domains containing $C_{0}^{\infty}(\mathbb{G})$, and such that $C_{0}^{\infty}(\mathbb{G})$ is an invariant subspace for them, that is,
$\mathcal{P}(C_{0}^{\infty}(\mathbb{G}))\subset C_{0}^{\infty}(\mathbb{G})$
and 
$\mathcal{M}(C_{0}^{\infty}(\mathbb{G}))\subset C_{0}^{\infty}(\mathbb{G})$.
The main (and only) assumption in this paper is that $\mathcal{P}$ and $\mathcal{M}$ satisfy the relations
\begin{equation}\label{PMEf}
2\,{\rm Re}\left(\mathcal{P}f\overline{i\mathcal{M}f}\right)
=(\mathcal{P}\circ(i\mathcal{M}))|f|^{2}=E|f|^{2}
\end{equation}
for all $f\in C_{0}^{\infty}(\mathbb{G})$.

In particular, in view of equalities \eqref{PMEfex3} in Example \ref{REM:pm}, it is satisfied by the operators 
$\mathcal{P}$ and $\mathcal{M}$ given in \eqref{EQ:pmex1}. 
However, surprisingly, we do not need their precise expressions from \eqref{EQ:pmex1} to derive subsequent properties presented in this paper: only the relation \eqref{PMEf} is required for our further analysis.

We denote by $D(\mathcal{P})$ and $D(\mathcal{M})$ the domains of operators 
$\mathcal{P}$ and $\mathcal{M}$, respectively.

\subsection{Position-momentum ($\mathcal{P}\mathcal{M}$) relations.} 
In this subsection we 
show relations between abstract position $\mathcal{P}$ and momentum $\mathcal{M}$ operators on homogeneous groups satisfying equalities \eqref{PMEf}.

\begin{thm}\label{Kennardequality}
Let $\mathbb{G}$ be a homogeneous group
of homogeneous dimension $Q$.
Then for every $f\in D(\mathcal{P})\bigcap D(\mathcal{M})$ with $\mathcal{P}f\not\equiv0$ and $\mathcal{M}f\not\equiv0,$
we have
\begin{multline}\label{awH}
\|\mathcal{P}f\|^{2}_{L^{2}(\mathbb{G})}+\|\mathcal{M}f\|^{2}_{L^{2}(\mathbb{G})}
=Q\|f\|^{2}_{L^{2}(\mathbb{G})}+\|\mathcal{P}f-i\mathcal{M}f\|^{2}_{L^{2}(\mathbb{G})}
\\=\|\mathcal{P}f\|_{L^{2}(\mathbb{G})}\|\mathcal{M}f\|_{L^{2}(\mathbb{G})}
\left(2- \left\|\frac{\mathcal{P}f}{\|\mathcal{P}f\|_{L^{2}(\mathbb{G})}}+\frac{i\mathcal{M}f}{\|\mathcal{M}f\|_{L^{2}(\mathbb{G})}}
\right\|^{2}_{L^{2}(\mathbb{G})}\right) \\ +\|\mathcal{P}f+i\mathcal{M}f\|^{2}_{L^{2}(\mathbb{G})}.
\end{multline}
\end{thm}

\begin{proof}[Proof of Theorem \ref{Kennardequality}]
There is a (unique)
positive Borel measure $\sigma$ on the
unit quasi-sphere
\begin{equation}\label{EQ:sphere}
\wp:=\{x\in \mathbb{G}:\,|x|=1\},
\end{equation}
such that for all functions $f\in L^{1}(\mathbb{G})$ we have the polar decomposition
\begin{equation}\label{EQ:polar}
\int_{\mathbb{G}}f(x)dx=\int_{0}^{\infty}
\int_{\wp}f(D_{r}(y))r^{Q-1}d\sigma(y)dr.
\end{equation}
We refer to Folland and Stein \cite{FS-Hardy} for the proof (see also
\cite[Section 3.1.7]{FR}).
Since $C_{0}^{\infty}(\mathbb{G})$ is dense in $L^{2}(\mathbb{G})$, we need to show \eqref{awH} for $f\in C_{0}^{\infty}(\mathbb{G})$ and then this implies that it is also true on
$D(\mathcal{P})\bigcap D(\mathcal{M})$ by density.
Using the above polar decomposition, formula \eqref{dfdr} and equality \eqref{PMEf}, we calculate
\begin{align*}
-2\,{\rm Re}\int_{\mathbb{G}}\mathcal{P}f\overline{i\mathcal{M}f}dx
& =-\int_{\mathbb{G}}\mathcal{P}i\mathcal{M}|f|^{2}dx=-\int_{0}^{\infty}\int_{\wp}
r^{Q} \frac{1}{r}E |f|^{2} d\sigma(y)dr
\\&=-\int_{0}^{\infty}\int_{\wp}
r^{Q} \frac{d|f|^{2}}{dr} d\sigma(y)dr=Q\int_{0}^{\infty}\int_{\wp}
r^{Q-1} |f|^{2}d\sigma(y)dr\\&=Q\int_{\mathbb{G}}|f|^{2}dx=Q\|f\|^{2}_{L^{2}(\mathbb{G})}.
\end{align*}
Combining this with
\begin{equation*}
\|\mathcal{P}f\|^{2}_{L^{2}(\mathbb{G})}+
\|\mathcal{M}f\|^{2}_{L^{2}(\mathbb{G})}
=\|\mathcal{P}f+i\mathcal{M}f\|^{2}_{L^{2}(\mathbb{G})}-2\,{\rm Re}\int_{\mathbb{G}}\mathcal{P}f\overline{i\mathcal{M}f}dx
\end{equation*}
we obtain the first equality in \eqref{awH}.
On the other hand, we have
$$-2\,{\rm Re}\int_{\mathbb{G}}\mathcal{P}f\overline{i\mathcal{M}f}dx
=\|\mathcal{M}f\|_{L^{2}(\mathbb{G})}\|\mathcal{P}f\|_{L^{2}(\mathbb{G})}\left(2- \left\|\frac{\mathcal{P}f}{\|\mathcal{P}f\|_{L^{2}(\mathbb{G})}}+
\frac{i\mathcal{M}f}{\|\mathcal{M}f\|_{L^{2}(\mathbb{G})}}
\right\|^{2}_{L^{2}(\mathbb{G})}\right).$$

This proves the second equality in \eqref{awH}.
\end{proof}

Equalities \eqref{awH} imply the following Heisenberg-Kennard inequality:
\begin{cor}\label{relations1}We have
\begin{equation}\label{Ken}
\frac{Q}{2}\|f\|^{2}_{L^{2}(\mathbb{G})}\leq\|\mathcal{P}f\|_{L^{2}(\mathbb{G})}\|\mathcal{M}f\|_{L^{2}(\mathbb{G})}
,\end{equation}
which is also called the Kennard uncertainty inequality in the abelian case (see e.g. \cite{SZ} and \cite{WM}).
\end{cor}
The first equality in \eqref{awH} implies the following Pythagorean type inequality:
\begin{cor}\label{relations2} We have
\begin{equation}
\|\sqrt{Q}f\|^{2}_{L^{2}(\mathbb{G})}\leq \|\mathcal{P}f\|^{2}_{L^{2}(\mathbb{G})}+\|\mathcal{M}f\|^{2}_{L^{2}(\mathbb{G})}
.\end{equation}
\end{cor}

Equalities \eqref{awH} also imply the following:
\begin{cor}\label{relations3}
\begin{itemize}
\item[(i)] Let $f\in D(\mathcal{P})\bigcap D(\mathcal{M})$ with $\mathcal{P}f\not\equiv0$ and $\mathcal{M}f\not\equiv0$. Then the equality case in the Heisenberg-Kennard uncertainty inequality \eqref{Ken} holds, that is,
$$\frac{Q}{2}\|f\|^{2}_{L^{2}(\mathbb{G})}=\|\mathcal{P}f\|_{L^{2}(\mathbb{G})}\|\mathcal{M}f\|_{L^{2}(\mathbb{G})}$$
if and only if $$\|\mathcal{P}f\|_{L^{2}(\mathbb{G})}i\mathcal{M}f=\|\mathcal{M}f\|_{L^{2}(\mathbb{G})}\mathcal{P}f.$$
\item[(ii)] For
$f\in D(\mathcal{P})\bigcap D(\mathcal{M})$ we have the Pythagorean equality
$$
\|\sqrt{Q}f\|^{2}_{L^{2}(\mathbb{G})}=\|\mathcal{P}f\|^{2}_{L^{2}(\mathbb{G})}+\|\mathcal{M}f\|^{2}_{L^{2}(\mathbb{G})}$$
if and only if $$\mathcal{P}f=i\mathcal{M}f.$$

\end{itemize}
\end{cor}

\subsection{Euler-Coulomb ($E\mathcal{C}$) relations.} Euler and Coulomb potential operators can be defined by
\begin{equation}Ef:=|x|\mathcal{R}f\end{equation}
and
\begin{equation}\mathcal{C}f:=\frac{1}{|x|}f,\end{equation}
with the corresponding domains
\begin{equation}
D(E)=\{f\in L^{2}(\mathbb{G}):\,Ef\in L^{2}(\mathbb{G})\}
\end{equation}
and
\begin{equation}
D(\mathcal{C})=\{f\in L^{2}(\mathbb{G}):\, \frac{1}{|x|}f\in L^{2}(\mathbb{G})\}.
\end{equation}
The Euler operator $E$ defines the homogeneity on $\mathbb G$: a $C^{1}$-function $f$ satisfies $f(\lambda x)=\lambda^{\mu}f(x)$ for all $\lambda>0$ if and only if $Ef=\mu f$.

The combination of the Euler operator and Coulomb potential defines an (radial derivative) operator $\mathcal{R}$ by the formula
\begin{equation}\label{EQ:Euler}
\mathcal{R} :=\mathcal{C}E,
\end{equation}
see \eqref{dfdr}.
Moreover, for each $f\in C^{\infty}_{0}(\mathbb{G}\backslash\{0\})$
one has (see \cite[Theorem 4.1]{Ruzhansky-Suragan:identities})
\begin{multline}\label{aR}
\left\|\frac{1}{|x|^{\alpha}}\mathcal{R} f\right\|^{2}_{L^{2}(\mathbb{G})}= \\
\left(\frac{Q-2}{2}-\alpha\right)^{2}
\left\|\frac{f}{|x|^{\alpha+1}}\right\|^{2}_{L^{2}(\mathbb{G})}
+\left\|\frac{1}{|x|^{\alpha}}\mathcal{R} f+\frac{Q-2-2\alpha}{2|x|^{\alpha+1}}f
\right\|^{2}_{L^{2}(\mathbb{G})},
\end{multline}
for all $\alpha\in\mathbb{R}.$

From \eqref{aR} one can get different inequalities, for example, by dropping the second positive term in the right hand side of \eqref{aR} (of course, one can obtain other inequalities by dropping the first term of the right hand side).

\begin{rem}\label{REM:Eucl1}
In the abelian case ${\mathbb G}=(\mathbb R^{n},+)$, $n\geq 3$, we have
$Q=n$, so for any homogeneous quasi-norm $|\cdot|$ on $\mathbb R^{n}$ \eqref{aR} implies
a new inequality with the optimal constant:
\begin{equation}\label{Hardy-r}
\frac{|n-2-2\alpha|}{2}\left\|\frac{f}{|x|^{\alpha+1}}\right\|_{L^{2}(\mathbb{R}^{n})}\leq
\left\|\frac{1}{|x|^{\alpha}}\frac{x}{|x|}\cdot\nabla f\right\|_{L^{2}(\mathbb{R}^{n})},\;\forall\alpha\in\mathbb{R},
\end{equation}
which in turn, by using Schwarz's inequality with the standard Euclidean distance $\|x\|=\sqrt{x^{2}_{1}+\ldots+x^{2}_{n}}$, implies the $L^{2}$ Caffarelli-Kohn-Nirenberg inequality \cite{CKN-1984} for $\mathbb{G}\equiv\mathbb{R}^{n}$ with the optimal constant:
\begin{equation}\label{CKN}
\frac{|n-2-2\alpha|}{2}\left\|\frac{f}{\|x\|^{\alpha+1}}\right\|_{L^{2}(\mathbb{R}^{n})}\leq
\left\|\frac{1}{\|x\|^{\alpha}}\nabla f\right\|_{L^{2}(\mathbb{R}^{n})},\;\forall\alpha\in\mathbb{R},
\end{equation}
for all $f\in C_{0}^{\infty}(\mathbb{R}^{n}\backslash\{0\}).$
Here optimality of the constant $\frac{|n-2-2\alpha|}{2}$ was proved in \cite[Theorem 1.1. (ii)]{CW-2001}.
\end{rem}

We now continue with general homogeneous groups $\mathbb{G}.$ 
If $\alpha=0$ from \eqref{aR} we obtain the equality
\begin{equation}\label{47-0}
\left\|\mathcal{R} f\right\|^{2}_{L^{2}(\mathbb{G})}=\left(\frac{Q-2}{2}\right)^{2}\left\|\frac{1}{|x|}f
\right\|^{2}_{L^{2}(\mathbb{G})}+\left\|\mathcal{R} f+\frac{Q-2}{2|x|}f
\right\|^{2}_{L^{2}(\mathbb{G})}.
\end{equation}
Now by dropping the nonnegative last term in \eqref{47-0} we immediately obtain a version of Hardy's inequality on $\mathbb{G}$ (see \cite{Ruzhansky-Suragan:identities} 
for its weighted $L^{p}$ version): 
\begin{equation}\label{47-1}
\left\|\frac{1}{|x|}f
\right\|_{L^{2}(\mathbb{G})}\leq\frac{2}{Q-2}\left\|\mathcal{R} f\right\|_{L^{2}(\mathbb{G})},\; Q\geq3.
\end{equation}
Note that in comparison to stratified (Carnot) group versions, here the constant is best
for any quasi-norm $|\cdot|$.

\begin{rem}\label{REM:Eucl2}
In the abelian case ${\mathbb G}=(\mathbb R^{n},+)$, $n\geq 3$, we have
$Q=n$, so for any homogeneous quasi-norm $|\cdot|$ on $\mathbb R^{n}$ it implies
the inequality
\begin{equation}\label{Hardy-r}
\left\|\frac{f}{|x|}\right\|_{L^{2}(\mathbb{R}^{n})}\leq \frac{2}{n-2}
\left\|\frac{x}{|x|}\cdot\nabla f\right\|_{L^{2}(\mathbb{R}^{n})},
\end{equation}
which in turn, again by using Schwarz's inequality with the standard Euclidean distance $\|x\|=\sqrt{x^{2}_{1}+\ldots+x^{2}_{n}}$, implies the classical Hardy inequality for $\mathbb{G}\equiv\mathbb{R}^{n}$:
\begin{equation*}\label{Hardy}
\left\|\frac{f}{\|x\|}\right\|_{L^{2}(\mathbb{R}^{n})}\leq
\frac{2}{n-2}\left\|\nabla f\right\|_{L^{2}(\mathbb{R}^{n})},
\end{equation*}
for all $f\in C_{0}^{\infty}(\mathbb{R}^{n}\backslash\{0\}).$

We also refer to a
recent interesting paper of Hoffmann-Ostenhof and Laptev
\cite{Laptev15} on this subjects
for Hardy inequalities with homogeneous weights, to \cite{HHLT-Hardy-many-particles} for many-particle versions and to many further references therein.
\end{rem}

By standard argument the inequality \eqref{47-1} implies the following Heisenberg-Pauli-Weyl type  uncertainly principle on homogeneous groups (see e.g. \cite{GL}, \cite{Ricci07}, \cite{Ricci15}, \cite{Ruzhansky-Suragan:Layers} for versions on abelian and stratified groups):

\begin{prop}\label{Luncertainty}
Let $\mathbb{G}$ be a homogeneous group of homogeneous dimension
 $Q\geq 3$.
Then for each $f\in C^{\infty}_{0}(\mathbb{G}\backslash\{0\})$ and any homogeneous quasi-norm $|\cdot|$ on $\mathbb{G}$ we have
\begin{equation}\label{UP1}
\left\|f\right\|^{2}_{L^{2}(\mathbb{G})}
\leq\frac{2}{Q-2}\left\|\mathcal{R} f\right\|_{L^{2}(\mathbb{G})}\left\||x| f\right\|_{L^{2}(\mathbb{G})}.
\end{equation}
\end{prop}
\begin{proof}
From the inequality \eqref{47-1} we get
$$
\left(\int_{\mathbb{G}}\left|\mathcal{R}f\right|^{2}dx\right)^{\frac{1}{2}}\left(\int_{\mathbb{G}}|x|^{2}
|f|^{2}dx\right)^{\frac{1}{2}}\geq$$
 $$\frac{Q-2}{2}\left(\int_{\mathbb{G}}
\frac{|f|^{2}}{|x|^{2}}\,dx\right)^{\frac{1}{2}}\left(\int_{\mathbb{G}}|x|^{2}
|f|^{2}dx\right)^{\frac{1}{2}}
\geq\frac{Q-2}{2}\int_{\mathbb{G}}
|f|^{2}dx,$$
where we have used the H\"older inequality in the last line.
This shows \eqref{UP1}.
\end{proof}

\begin{rem}\label{REM:Eucl3}
In the abelian case ${\mathbb G}=(\mathbb R^{n},+)$, we have
$Q=n$, so that \eqref{UP1} implies
the uncertainly principle with any quasi-norm $|x|$:
\begin{equation}\label{UPRn-r}
\left(\int_{\mathbb R^{n}}
 |u(x)|^{2} dx\right)^{2}\leq\left(\frac{2}{n-2}\right)^{2}\int_{\mathbb R^{n}}\left|\frac{x}{|x|}\cdot\nabla u(x)\right|^{2}dx
\int_{\mathbb R^{n}} |x|^{2} |u(x)|^{2}dx,
\end{equation}
which in turn implies
the classical
uncertainty principle for $\mathbb{G}\equiv\mathbb R^{n}$ with the standard Euclidean distance $\|x\|$:
\begin{equation*}\label{UPRn}
\left(\int_{\mathbb R^{n}}
 |u(x)|^{2} dx\right)^{2}\leq\left(\frac{2}{n-2}\right)^{2}\int_{\mathbb R^{n}}|\nabla u(x)|^{2}dx
\int_{\mathbb R^{n}} \|x\|^{2} |u(x)|^{2}dx,
\end{equation*}
which is the Heisenberg-Pauli-Weyl uncertainly principle on $\mathbb R^{n}$.
\end{rem}

Moreover, we have the following Pythagorean relation for the Euler operator:
\begin{prop} We have
\begin{equation}\label{ECrel}
\left\|E f\right\|^{2}_{L^{2}(\mathbb{G})}=\left\|\frac{Q}{2}f
\right\|^{2}_{L^{2}(\mathbb{G})}+\left\|E f+\frac{Q}{2}f
\right\|^{2}_{L^{2}(\mathbb{G})}
\end{equation}
for any $f\in D(E)$.
\end{prop}

\begin{proof}
Taking $\alpha=-1$, from \eqref{aR} we obtain
\eqref{ECrel} for any $f\in C_{0}^{\infty}(\mathbb{G}\backslash\{0\})$. Since $D(E)\subset L^{2}(\mathbb{G})$ and $C_{0}^{\infty}(\mathbb{G}\backslash\{0\})$ is dense in $L^{2}(\mathbb{G})$, this implies that \eqref{ECrel} is also true on
$D(E)$ by density.
\end{proof}

Simply by dropping the positive term in the right hand side,  \eqref{ECrel} implies
\begin{cor} We have
\begin{equation}
\left\|f
\right\|_{L^{2}(\mathbb{G})}\leq\frac{2}{Q}\left\|E f\right\|_{L^{2}(\mathbb{G})},
\end{equation}
for any $f\in D(E)$.
\end{cor}

\subsection{Radial-dilations-Coulomb ($\mathcal{R}\mathcal{R}_{g}\mathcal{C}$) relations.} A generator of dilations operator can be defined by
\begin{equation}
\mathcal{R}_{g}:=-i\left(\mathcal{R}+\frac{Q-1}{2}C\right)
\end{equation}
with the domain
\begin{equation}D(\mathcal{R}_{g})=\{f\in L^{2}(\mathbb{G}):\, \mathcal{R}f\in L^{2}(\mathbb{G}),\, Cf\in L^{2}(\mathbb{G})\}.\end{equation}
Note that the generator of dilations operator $\mathcal{R}_{g}$ and the Coulomb potential operator have the following special commutation relation
\begin{lem}\label{lemRgC1}For any $f\in C_{0}^{\infty}(\mathbb{G}\backslash\{0\})$ we have
\begin{equation}\label{comrelRsC}
[\mathcal{R}_{g}, \mathcal{C}]f=i\mathcal{C}^{2}f,
\end{equation}
where $[\mathcal{R}_{g}, \mathcal{C}]=\mathcal{R}_{g}\mathcal{C}-\mathcal{C}\mathcal{R}_{g}.$
\end{lem}
\begin{proof} [Proof of Lemma \ref{lemRgC1}]
Denoting
$r:=|x|$ we have
$\mathcal{C}=\frac{1}{r},$
 and from \eqref{dfdr} it follows that $\mathcal{R}_{g}=-i\left(\frac{d}{dr}+\frac{Q-1}{2r}\right).$
Thus, a direct calculation shows
 $$[\mathcal{R}_{g}, \mathcal{C}]f=\mathcal{R}_{g}\mathcal{C}f-\mathcal{C}\mathcal{R}_{g}f$$
 $$=-i\left(-\frac{1}{r^{2}}+\frac{1}{r}\frac{d}{dr}+\frac{Q-1}{2r^{2}}-\frac{1}{r}\frac{d}{dr}
 -\frac{Q-1}{2r^{2}}\right)f=i\frac{1}{r^{2}}f=i\mathcal{C}^{2}f.$$
\end{proof}

\begin{lem}\label{lemRgC2}
Operators $\mathcal{R}_{g}$ and $\mathcal{C}$ are symmetric.
\end{lem}
\begin{proof}[Proof of Lemma \ref{lemRgC2}]
It is a straightforward that $\mathcal{C}$ is symmetric, that is,
$$\int_{\mathbb{G}}(\mathcal{C}f)\overline{f}dx=\int_{\mathbb{G}}f\overline{\mathcal{C}f}dx.$$
Now we need to show that
\begin{equation}\label{symmetry}
\int_{\mathbb{G}}(\mathcal{R}_{g}f)\overline{f}dx=\int_{\mathbb{G}}f\overline{\mathcal{R}_{g}f}dx
\end{equation}
for any $f\in C_{0}^{\infty}(\mathbb{G}\backslash \{0\})$. Since $D(\mathcal{R}_{g})\subset L^{2}(\mathbb{G})$ and $C_{0}^{\infty}(\mathbb{G}\backslash \{0\})$ is dense in $L^{2}(\mathbb{G})$ it follows that \eqref{symmetry} is also true on
$D(\mathcal{R}_{g})$ by density if it is valid on $C_{0}^{\infty}(\mathbb{G}\backslash \{0\})$.
Using the polar decomposition with $\mathcal{R}_{g}=-i\left(\frac{d}{dr}+\frac{Q-1}{2r}\right)$ we obtain
$$\int_{\mathbb{G}}(\mathcal{R}_{g}f)\overline{f}dx=-i\int_{0}^{\infty}\int_{\wp}
r^{Q-1}\left(\frac{df}{dr}+\frac{Q-1}{2r}f\right)\overline{f}d\sigma(y)dr$$
$$
=-i\int_{0}^{\infty}\int_{\wp}
\frac{df}{dr}\overline{f}r^{Q-1}d\sigma(y)dr-i\frac{Q-1}{2}\int_{0}^{\infty}\int_{\wp}r^{Q-1}\frac{f}{r}\overline{f}d\sigma(y)dr
$$
$$
=i\int_{0}^{\infty}\int_{\wp}
f\frac{d\overline{f}}{dr}r^{Q-1}d\sigma(y)dr
$$
$$
+i(Q-1)\int_{0}^{\infty}\int_{\wp}r^{Q-1}\frac{f}{r}\overline{f}d\sigma(y)dr
-i\frac{Q-1}{2}\int_{0}^{\infty}\int_{\wp}r^{Q-1}\frac{f}{r}\overline{f}d\sigma(y)dr
$$
$$=\int_{0}^{\infty}\int_{\wp}
r^{Q-1}f\overline{\left(-i\frac{df}{dr}-i\frac{Q-1}{2r}f\right)}d\sigma(y)dr
=\int_{\mathbb{G}}f\overline{\mathcal{R}_{g}f}d\nu,$$
proving that $\mathcal{R}_{g}$ is also symmetric.
\end{proof}

For any symmetric operators $A$ and $B$ in $L^{2}$ with domains $D(A)$ and $D(B),$ respectively, a straightforward calculation (see e.g. \cite[Theorem 2.1]{OY-uncertainty}) shows the equality
\begin{multline}\label{2.10-}
-i\int_{\mathbb{G}}([A,B]f)\overline{f}d\nu
\\=\|Af\|_{L^{2}(\mathbb{G})}\|Bf\|_{L^{2}(\mathbb{G})}
\left(2- \left\|\frac{Af}{\|Af\|_{L^{2}(\mathbb{G})}}+i
\frac{Bf}{\|Bf\|_{L^{2}(\mathbb{G})}}
\right\|^{2}_{L^{2}(\mathbb{G})}\right),
\end{multline}
for $f\in D(A)\cap D(B)$ with $Af\not\equiv0$ and $Bf\not\equiv0$,
which will be useful in our next proof.

\begin{thm}\label{RsRCrelation}
Let $\mathbb{G}$ be a homogeneous group
of homogeneous dimension $Q$.
Then for every $f\in D(\mathcal{R})\cap D(\mathcal{C})$
we have
\begin{equation}\label{RsRC}
\|\mathcal{R}f\|^{2}_{L^{2}(\mathbb{G})}=\|\mathcal{R}_{g}f\|^{2}_{L^{2}(\mathbb{G})}+
\frac{(Q-1)(Q-3)}{4}\|\mathcal{C}f\|^{2}_{L^{2}(\mathbb{G})},
\end{equation}
and
\begin{equation}\label{RsC}
\|\mathcal{C}f\|_{L^{2}(\mathbb{G})}=
\|\mathcal{R}_{g}f\|_{L^{2}(\mathbb{G})}
\left(2- \left\|\frac{\mathcal{R}_{g}f}{\|\mathcal{R}_{g}f\|_{L^{2}(\mathbb{G})}}+i
\frac{\mathcal{C}f}{\|\mathcal{C}f\|_{L^{2}(\mathbb{G})}}
\right\|^{2}_{L^{2}(\mathbb{G})}\right),
\end{equation}
for $\mathcal{R}_{g}f\not\equiv0$ and $\mathcal{C}f\not\equiv0.$
\end{thm}

\begin{proof}[Proof of Theorem \ref{RsRCrelation}]
As in the proof of Theorem \ref{Kennardequality} we can calculate

\begin{multline*}
\|\mathcal{R}_{g}f\|^{2}_{L^{2}(\mathbb{G})}
=\left\|\mathcal{R} f+\frac{Q-1}{2|x|}f\right\|^{2}_{L^{2}(\mathbb{G})}
\\
=\left\|\mathcal{R} f\right\|^{2}_{L^{2}(\mathbb{G})}+(Q-1)\,{\rm Re}\int_{\mathbb{G}}\left(\mathcal{R} f\right)\overline{\frac{1}{|x|}f}dx+\left\|
\frac{Q-1}{2|x|}f\right\|^{2}_{L^{2}(\mathbb{G})}
\\
=\left\|\mathcal{R} f\right\|^{2}_{L^{2}(\mathbb{G})}+(Q-1)\,{\rm Re}\int_{0}^{\infty}\int_{\wp}
r^{Q-1}\left(\frac{d}{dr} f\right)\overline{\frac{1}{r}f}d\sigma(y)dr+\left\|
\frac{Q-1}{2|x|}f\right\|^{2}_{L^{2}(\mathbb{G})}
\\
=\left\|\mathcal{R} f\right\|^{2}_{L^{2}(\mathbb{G})}+
\frac{Q-1}{2}\int_{0}^{\infty}\int_{\wp}
r^{Q-2} \frac{d}{dr}|f|^{2} d\sigma(y)dr+\frac{(Q-1)^{2}}{4}\left\|
\mathcal{C}f\right\|^{2}_{L^{2}(\mathbb{G})}
\\=\left\|\mathcal{R} f\right\|^{2}_{L^{2}(\mathbb{G})}-\frac{(Q-1)(Q-2)}{2}\int_{0}^{\infty}\int_{\wp}
r^{Q-1}\frac{1}{r^{2}}|f|^{2} d\sigma(y)dr+\frac{(Q-1)^{2}}{4}\left\|
\mathcal{C}f\right\|^{2}_{L^{2}(\mathbb{G})}
\\=\left\|\mathcal{R} f\right\|^{2}_{L^{2}(\mathbb{G})}-\frac{(Q-1)(Q-2)}{2}\int_{\mathbb{G}}
|\mathcal{C}f|^{2} dx+\frac{(Q-1)^{2}}{4}\left\|
\mathcal{C}f\right\|^{2}_{L^{2}(\mathbb{G})}
\\=\left\|\mathcal{R} f\right\|^{2}_{L^{2}(\mathbb{G})}-\frac{(Q-1)(Q-3)}{4}\left\|
\mathcal{C}f\right\|^{2}_{L^{2}(\mathbb{G})}.
\end{multline*}
This proves \eqref{RsRC}.
Using \eqref{comrelRsC} and Lemma \ref{lemRgC2} with \eqref{2.10-} we obtain
\begin{multline*}
\|\mathcal{C}f\|^{2}_{L^{2}(\mathbb{G})}=-i\int_{\mathbb{G}}[\mathcal{R}_{g}, \mathcal{C}]f\overline{f}dx\\=
\|\mathcal{R}_{g}f\|_{L^{2}(\mathbb{G})}\|\mathcal{C}f\|_{L^{2}(\mathbb{G})}
\left(2- \left\|\frac{\mathcal{R}_{g}f}{\|\mathcal{R}_{g}f\|_{L^{2}(\mathbb{G})}}+i
\frac{\mathcal{C}f}{\|\mathcal{C}f\|_{L^{2}(\mathbb{G})}}
\right\|^{2}_{L^{2}(\mathbb{G})}\right).
\end{multline*}
 As above since $C_{0}^{\infty}(\mathbb{G})$ is dense in $L^{2}(\mathbb{G})$, it implies that this equality is also true on
$D(\mathcal{R})\cap D(\mathcal{C})$ by density.
\end{proof}
The equality \eqref{RsRC} implies that
\begin{cor}
Let $Q\geq3$. The generator of dilations and Coulomb potential operator are bounded by the (radial) operator $\mathcal{R}$, that is,
\begin{equation}
\|\mathcal{R}_{g}f\|_{L^{2}(\mathbb{G})}\leq\|\mathcal{R}f\|_{L^{2}(\mathbb{G})},
\end{equation}
and
\begin{equation}
\frac{\sqrt{(Q-1)(Q-3)}}{2}\|\mathcal{C}f\|_{L^{2}(\mathbb{G})}\leq
\|\mathcal{R}f\|_{L^{2}(\mathbb{G})}.
\end{equation}
\end{cor}

The equality \eqref{RsC} implies that
\begin{cor}
The Coulomb potential operator is bounded by the generator of dilations operator with relative bound 2, that is,
\begin{equation}
\|\mathcal{C}f\|_{L^{2}(\mathbb{G})}\leq
2\|\mathcal{R}_{g}f\|_{L^{2}(\mathbb{G})}.
\end{equation}
\end{cor}

\medskip

{\bf Acknowledgement.} We are grateful to Professor Tohru Ozawa for providing
inspiration for the present work.

\end{document}